\documentclass[11pt]{article}
\usepackage[a4paper,hmargin=1in,vmargin=1.25in]{geometry}
\usepackage{amsmath,amsfonts,amssymb,amsthm,url,color,algorithm2e,bbm}
\usepackage{graphicx}
\usepackage{multirow}
\usepackage{tikz}
\usetikzlibrary{backgrounds}
\usepackage[affil-it]{authblk}

\def\Par{\mathcal{P}} 
\def\Parp{\Par(P)}
\def\Forb{\mathcal{A}}  
\def\Auth{\mathcal{B}} 
\def\Cov{\mathcal{C}}  
\def\Fam{\mathcal{B}} 

\def\s{\setminus}
\def\cl{\mathrm{cl}}

\theoremstyle{plain}
\newtheorem{proposition}{Proposition}
\newtheorem{theorem}[proposition]{Theorem}

\newtheorem{lemma}[proposition]{Lemma}

\theoremstyle{definition}
\newtheorem{definition}[proposition]{Definition}

\newtheorem{example}[proposition]{Example}

\newtheorem{remark}[proposition]{Remark}

\begin{document}
\title{Privacy-preserving Data Splitting: A Combinatorial Approach}
\author{Oriol Farr\`as, Jordi Ribes-Gonz\'alez, Sara Ricci\\%
\texttt{\{oriol.farras,jordi.ribes,sara.ricci\}@urv.cat}}
\affil{Department of Mathematics and Computer Science, \\Universitat Rovira i Virgili,\\ Tarragona, Spain.}
\date{\today}


\maketitle

\begin{abstract}
Privacy-preserving data splitting is a technique that aims to protect data privacy by storing different fragments of data in different locations. 
In this work we give a new combinatorial formulation to the data splitting problem. We see the data splitting problem as a purely combinatorial problem, in which we have to split data attributes into different fragments in a way that satisfies certain combinatorial properties derived from processing and privacy constraints. Using this formulation, we develop new combinatorial and algebraic techniques to obtain solutions to the data splitting problem. We present an algebraic method which builds an optimal data splitting solution by using Gr\"{o}bner bases. Since this method is not efficient in general, we also develop a greedy algorithm for finding solutions that are not necessarily minimal sized.  
\end{abstract}

\section{Introduction}


Data collected by companies and organizations is increasingly large and it is nowadays unfeasible for some data owners to
locally store and process it because of the associated costs (such as hardware, energy and maintenance costs). 
The cloud offers a suitable alternative for data storage, by providing large and highly scalable storage and computational resources 
at a low cost and with ubiquitous access. However, many data owners are reluctant to embrace the cloud computing technology 
because of security and privacy concerns, mainly regarding the cloud service provider (CSP). The problem is not only that 
CSPs may read, use or even sell the data outsourced by their customers; but also that they may suffer attacks or data leakages that can compromise data confidentiality.

Privacy-preserving data splitting is a technique that aims to protect data privacy in this setting. 
Data splitting minimizes the leakage of information by distributing the data among several CSPs,
assuming that they do not communicate with each other. Similar problems have been studied in other areas such as data mining, data sanitization, file splitting and data merging.

In general, in data splitting data sets are structured in a tabular format, according to a set of \emph{attributes} (or \emph{features}) identifiable by attribute names, as the table schema.
Data is then composed by \emph{records}, where each record holds up to one value per attribute. 
For instance, we can consider the attributes \verb|`Name',`Age',`Occupation'|, and a 
record \verb|{`John',`21',`Student'}|, where the 
record holds values for all attributes.

Data splitting comes in three flavours: horizontal, vertical and mixed. 
In this work, we deal with vertical data splitting, where fragments consist of data on all records, but only contain information on a subset of the attributes. In horizontal data splitting, fragments contain part of the records, and information on all attributes is specified. In mixed data splitting, fragments hold partial information on some records. 

Horizontal data splitting is not privacy-preserving by itself,
because all the information of an individual register is stored together; hence, it does not preserve privacy by decomposition \cite{GM}. Horizontal data splitting has been used to analyze data collected by different entities on a set of patients \cite{gaye2014datashield}, or in conjunction with homomorphic encryption \cite{kantarcioglu2008survey}, to mine horizontally-partitioned data without violating privacy.

Vertical data splitting can be used for privacy-preserving purposes~\cite{GM,GM-2}. In particular, in a setting where some combinations of attributes constitute the sensitive information, the data set can be vertically split and distributed among CSPs so that no CSP holds any sensitive attribute combination. Assuming that CSPs do not communicate with each other, this measure guarantees privacy. An example of a sensitive pair of attributes in a medical data setting is passport number and disease, whereas blood pressure and disease constitute a generally safe pair. 

The results we present in this work focus on data splitting, but they can be applied to other related areas such as file splitting, data sanitization, and data merging. 

In file splitting, pieces of files owned by the same entity are stored in different sites. This is done in such a way that pieces from each site, when considered in isolation, are not sensitive. 
In \cite{abu2010racs}, the authors spread the data across multiple CSPs and introduce redundancy in order to tolerate possible 
failures or outages. Their solution follows what has been done for years at the level of disks and file 
systems, particularly in the RAID (Redundant Array of Independent Disks) technology, which strips data across an array of 
disks and maintains parity data that can be used to reconstruct the contents of any individual failed disk. 
In~\cite{dev2012approach}, user's files are categorized and split into chunks, and these chunks are provided to the proper storage servers. 
The categorization of data is done according to mining sensitivity. To ensure a greater amount of privacy, 
the possibility of adding misleading data into chunks depending on the demand of clients is given. 
Wei et al.~\cite{wei2013data} proposed a new privacy method that involves bit splitting and bit combination. 
In their approach, the original files are broken up through bit splitting and each fragment is uploaded to a different storage server.

Data sanitization is the process of removing sensitive information from a document so as to ensure that only the intended information 
can be accessed. Typically, the result is a document that is suitable for dissemination to the intended audience. Data sanitization has 
been applied along with data splitting in \cite{clarus}, where the terms in the input document that cause disclosure risk according to
the privacy requirements are first detected, and then those terms are distributed in multiple servers in order to prevent disclosure. 

Data merging consists on securely splitting and merging data from potentially many sources in a single repository. An approach for data merging is to split and compress the data into multiple fragments, and to require certain privacy constraints on the fragments~\cite{brinkman2006lucky}.  

In data splitting, once the data is split, the main issue is how to securely compute over distributed data (see \cite{yang200610} for more details). 
For some computations the servers may need to exchange data, but none of them ought to reveal its own private stored information. 
Computing over distributed data is also studied in the context of parallel processing for statistical computations. In this case, the problem is how to combine partial results obtained from independent processors. 
Related literature reduces statistical analysis to 
performing either secure distributed scalar products or secure distributed matrix products, e.g. see \cite{du2004privacy,karr2009privacy,ricci2016privacy}. On a similar note, the field of privacy-preserving data mining deals with the problem of computing over distributed data. It has as its main objective to mine data owned by different parties, who are willing to collaborate in order to get better 
results, but that do not want or cannot share the raw original data. For instance, see \cite{calvino2015privacy,clifton2002tools,goethals2004private}.

%

\subsection{Our results}

In this work we give a new combinatorial formulation to the data splitting problem. 
In the considered data splitting problem, we force some subsets of attributes to be stored separately, because the combination of these attributes may reveal sensitive information to the CSPs. Moreover, we want to impose some subsets of attributes to be stored together, because we want to query on them efficiently or to compute statistics on them (data mining, selective correlations, etc.). Regarding privacy and security, the CSPs are not trusted and hence they are not given access to the entire original data set. We thus assume that the CSPs have just access to fragments of the original data set. 

More specifically, we consider the honest-but-curious security model, where the CSPs honestly fill their role in the protocols and do not share information with each other, but they may try to infer information on the data available to them. In particular, each CSP may be curious to analyze the data it stores and the message flows received during the protocol in order to acquire additional information. Therefore, in our model the information leakage is the sensitive information that can be extracted from single stored data fragments. This model is common in the cloud computing literature, e.g. see \cite{cao2014privacy}.

In this setting, our main objective is to minimize the number of CSPs that are needed to store a data set using data splitting, without applying other privacy-preserving techniques. In order to study this problem, we set the data splitting constraints as two families of subsets of attributes: the family of subsets of attributes that have to be stored together, and the family of subsets of attributes that must not be jointly stored in any CSP. These two families respectively define processing and privacy constraints.
We define a data splitting solution as a family of subsets of attributes which satisfies the processing and privacy constraints. Each set in this family must be outsourced to a single CSP. Therefore, we see the data splitting problem as a purely combinatorial problem, in which we have to split attributes into different fragments in a way that satisfies certain combinatorial properties derived from processing and privacy constraints. 

Using this formulation, we develop new combinatorial and algebraic techniques to obtain solutions to the data splitting problem. We first present an algebraic method which builds a data splitting solution with the minimal number of fragments by using Gr\"{o}bner bases. Since this method has performance issues, we also develop an efficient greedy algorithm for finding solutions that are not generally minimal sized. We compare the efficiency and the accuracy of the two approaches by giving experimental results. Using results of graph theory, we are able to provide necessary and sufficient conditions for the existence of a solution to the data splitting problem, and we give upper and lower bounds on the number of needed fragments.


\subsection{Related work}\label{sec:related_work}

Recently, data splitting research has focused on finding the minimal sized decomposition of a given data set into privacy-preserving fragments. Related works suggest outsourcing a sensitive data set by vertically splitting it according to some privacy constraints~\cite{GM,CirianiDFJPS:07,CCFJPS,ciriani2011selective,GM-2}. In all previously proposed methods, privacy constraints are described by sensitive pairs of attributes. 

In \cite{GM}, the authors study the problem of finding a decomposition of a given data set into two privacy-preserving vertical fragments, so as to store them in two CSPs which have to be completely unaware of each other. Query execution is also optimized, i.e. they minimize the execution cost of a given query workload, while obeying the constraints imposed by the needs of data privacy. Graph-coloring techniques are used to identify a decomposition with small query costs. In particular, their data splitting problem can be reformulated as a hypergraph-coloring instance of a graph $G$. In case some sensitive attribute pairs can not be stored separately without increasing the number of fragments to more than two, encryption is used to ensure privacy. To improve the query workload, the storage of the same attribute in both 
CSPs is allowed.

The optimal decomposition problem described in \cite{GM} is hard to solve even if vertex deletion is not allowed. In fact, Guruswami et al.~\cite{guruswami2002hardness} proved that it is NP-hard to color a 2-colorable, 4-uniform hypergraph using only $c$ colors for any constant $c$. This means that, in the case that all 4-tuples of attributes are sensitive, it is NP-hard to find a partition of attributes into two sets that satisfies all privacy constraints, even knowing that it exists. Because of the hardness of this problem, in \cite{GM} they present three different heuristics to solve it. 

A later article \cite{GM-2} studies the same scenario as \cite{GM}. Here as well, they consider vertically splitting data into exactly two fragments, though their results are easily extendable to more fragments. They also allow encrypting sensitive attributes and storing the same attribute in both fragments. 
%
They introduce three heuristics to find a decomposition with small query costs. These heuristic search techniques are based on the greedy hill climbing approach, and give a nearly optimal solution.

In \cite{GM-2}, the authors study the time complexity of the proposed optimal decomposition problem in terms of the number of attributes. 
The general problem can theoretically be solved in polynomial time if the collection contains only few sets of constraints (by solving the minimum cut problem). It can also be solved in logarithmic time $O(\log(n))$ when the problem is equivalent to the hitting set problem. And it can be solved in an approximation factor of $O(\sqrt{n})$ if each constraint set has size $2$, by using directed multicat (i.e., solving the minimum edge deletion bipartition problem). The problem becomes intractable if the sets of constraints have size $3$. In fact, in this case the problem is reduced to the not-all-equal $3$-satisfiability problem, which is an NP-complete problem. 

Also, \cite{ciriani2011selective} presents a solution for vertically splitting data into two fragments without requiring the use of encryption, but rather by using a trusted party (the data owner) to store a portion of the data and to perform part of the computation. 

The solution presented in \cite{CCFJPS} uses both encryption and data splitting, but it allows the CSPs to communicate between each other. Because of this assumption, in order to ensure unlinkability between attributes, no attribute must appear in the clear in more that one fragment. 
In their solution, data is split into possibly more than two different fragments. This lowers the complexity of the problem with respect to \cite{GM} and \cite{GM-2}. The optimization problem is then to find a partition that minimizes the number of fragments and maximizes the number of attributes stored in the clear. Also, in this case, the problem of finding a partition of the attribute set is NP-complete. Hence, they present two heuristic methods with time complexity $O(n^2m)$ and $O(n^3m)$, where $m$ is the number of privacy constraints and $n$ is the number of attributes. The first one is based on the definition of vector minimality, and the second one works with an affinity matrix that expresses the advantage of having pairs of attributes in the same fragment. 

A similar approach to \cite{CCFJPS} is also illustrated in \cite{CirianiDFJPS:07}, where they split a data set into an arbitrary number of non-linkable data fragments and distribute them among an arbitrary number of non-communicating servers.

The data splitting problem studied in this work is also related to other well known combinatorial optimization problems. We want to emphasize the connection with the job shop scheduling problem. The job shop scheduling problem consists on assigning jobs to resources at particular times. Welsh and Powell~\cite{WP} described a basic scheduling problem as follows: let $J=\{J_{i}\}_{i=1}^{N}$ be a set of $N$ jobs. Suppose that it takes an entire day to complete each job, and that resources are unbounded. Let $M=\{m_{ij}\}_{i,j=1}^{N}$ be an incompatibility matrix, where $m_{ij}$ is zero or one depending on whether or not $J_{i}$ and $J_{j}$ can be carried out on the same day. The problem consists in scheduling the $n$ jobs using the minimum needed number of days according to the restrictions imposed by $M$. An efficient algorithm to solve this problem is presented in \cite{WP}, and subsequent works \cite{L,B} improve on this solution. See \cite{C} for a survey on this and similar scheduling problems. 

By interpreting jobs as attributes, days as data locations and the incompatibility matrix as a set of privacy constraints, we observe the equivalence between the problem posed in \cite{WP} and the data splitting problem. Through this same analogy, our setting extends to the following job scheduling problem: let $J=\{J_{i}\}_{i=1}^{N}$ be a set of $N$ jobs, and suppose that it takes an entire day to complete each job, and that resources are unbounded. Let $\Forb\subseteq\Par(J)$ be a family of sets of jobs that can not be carried out all on the same day. Similarly, let $\Auth\subseteq\Par(J)$ be a family of sets of jobs that must be carried out all on the same day. The problem consists in scheduling the $n$ jobs using the minimum needed number of days according to the restrictions imposed by $\Forb$ and $\Auth$.

\subsection{Outline of the work}

Section \ref{sec:comb_approach} states the problem of privacy-preserving data splitting as a purely combinatorial problem which consists
in splitting sensitive data into several fragments. This data splitting problem is stated as a covering problem. 
Section \ref{sec:algebraic_form} presents an algebraic formulation of the covering problem stated in the previous section. 
Gr\"obner basis is used to find the optimal (i.e., minimal-sized) solutions.
Section \ref{section:effalgorithm} proposes a linear-time method which solves the combinatorial problem. The solution optimality 
has been sacrificed to benefit efficiency. 
A heuristic improvement is also proposed.
Section \ref{sec:exp_results} presents the experimental results obtained by implementing
the methods presented in Sections \ref{sec:algebraic_form} and \ref{section:effalgorithm}. First, a comparison between the methods on 
a real problem is depicted. Then, a performance analysis of the linear-time methods has been carried out over random graphs.   
Finally, Section \ref{sec:conclusions} lists some conclusions.

\section{A combinatorial approach}\label{sec:comb_approach}

In this section we state the problem of privacy-preserving data splitting as a purely combinatorial problem. 
This problem consists in splitting a given data set in which some attributes are sensitive. As discussed above, this situation also covers problems of file splitting, data sanitization and data merging.
First we introduce some notation.

Let $P$ be a set and let $\Cov\subseteq \Par(P)$. For any $i\in P$, we define $\deg_i(\Cov)$ as the number of subsets in $\Cov$ containing $i$, and we define the degree $\deg(\Cov)$ of $\Cov$ as the maximum of $\deg_i(\Cov)$ for every $i\in P$. For any $B\subseteq P$, we also denote by $\deg_{B}(\Cov)$ the number of subsets $A\in\Cov$ such that $A\cap B\neq\emptyset$. Note that for any $i\in P$ we have $\deg_{i}(\Cov)=\deg_{\{i\}}(\Cov)$. For a set $A\subseteq P$, we define its closure 
$\cl(A)=\{B\subseteq P\, :\, A\subseteq B\}$. 
We define $\min{\Cov}$ and $\max{\Cov}$ as follows. A subset $A\subseteq P$ is in $\min{\Cov}$ if and only if $A\in\Cov$ and it does not exist $B\in\Cov$ with $B\subsetneq A$. Analogously, a subset  $A\subseteq P$ is in $\max{\Cov}$ if and only if $A\in\Cov$ and it does not exist $B\in\Cov$ with $A\subsetneq B$. That is, a $\min{\Cov}$ is the family of minimal subsets in $\Cov$, and $\max{\Cov}$ is the family of maximal subsets in $\Cov$.  
We say that $\Cov$ is an \emph{antichain} if $A\nsubseteq B$ for every $A,B\in\Cov$. In this case, $\Cov=\min{\Cov}=\max{\Cov}$.


In the considered data splitting setting we have a set of attributes $P$, and some combinations of the attributes 
are not to be stored by any individual server because they would leak sensitive information. We assume that individual attributes, when considered in isolation, are not sensitive (otherwise, encryption can be used). Moreover, we want some other attributes to be stored in the same location,
for example to perform statistical analysis computations such as
contingency tables, correlations or principal component analysis of the attributes.   
We thus describe a {\em data splitting problem} using two families of attributes: 
$\Forb\subseteq \Par(P)$ is the family of subsets of attributes that cannot be stored together in any single server, 
and $\Auth\subseteq \Par(P)$ is the family of subsets of attributes that we want to be stored together in some server.
We state the data splitting problem in terms of {\em coverings}, a notion first introduced in~\cite{FRR}.

\begin{definition}
\label{def:cov}
Let $\Forb,\Auth\subseteq \Par(P)$.
An $(\Forb,\Auth)$-covering $\Cov$ is a family of subsets of $P$ satisfying that
\begin{enumerate}
\item for every $A\in \Forb$ and for every $C\in\Cov$, $A\nsubseteq C$, and
\item for every $B\in \Auth$ there exists $C\in\Cov$ with $B\subseteq C$. 
\end{enumerate}
\end{definition}

Let $\Forb,\Auth$ be the families of subsets defined by the data splitting restrictions described above, and let $\Cov$ be an $(\Forb,\Auth)$-covering. Then $\Cov$ defines a solution for data splitting by associating each fragment $i$ with a set $C_i\in\Cov$. 
That is, we solve the data splitting problem by storing the data corresponding to attributes in $C_i$ at the $i$-th location. Observe that, according to this definition, for each $B\in\Auth$  there is at least one fragment containing all attributes in $B$, and none of the fragments contain all attributes in $A$ for any $A\in\Forb$. These are exactly the restrictions we have for data splitting. Note that we distribute data in as many fragments as $|\Cov|$.
Since $\Auth\subseteq \Par(P)$ is the family of subsets of attributes that we want to be stored together, we will always assume that $\cup_{B\in \Auth} B=P$. 

Our work is focused on minimizing the size of the coverings, which corresponds to the number of fragments in data splitting. Therefore, we say that $\Cov$ is an \emph{optimal}  $(\Forb,\Auth)$-covering if $|\Cov|$ is minimal among all $(\Forb,\Auth)$-coverings. 
Also, it could be desirable to minimize $\sum_{X\in \Cov} |X|$,  which corresponds to the total amount of information that will be stored, and $\max_{i\in P} \deg_i(\Cov)$, which corresponds to the maximum redundancy in the storage. 

\begin{example}
\label{e:cA}
Let $\Fam \subseteq \Par(P)$ be an antichain, and let $A\in\Fam$. Then $\Cov=\{P\s\{i\}\, :\, i\in A\}$ is a $(\{A\},\Fam\s\{A\})-covering$. 
\end{example}

Next we present some technical results about coverings. The main results of this section are Proposition~\ref{p:ecov}, which characterizes the existence of coverings, and Proposition~\ref{p:anti}, which justifies the search of $(\Forb,\Auth)$-coverings in the case that $\Forb$ and $\Auth$ are antichains. In addition, we present a theoretical lower bound on the size of $(\Forb,\Auth)$-coverings in Proposition~\ref{p:lowerbound}.

\begin{lemma}
\label{l:tcov}
Let  $\Forb,\Auth,\Cov\subseteq \Par(P)$. Then
$\Cov$ is an $(\Forb,\Auth)$-covering if and only if
\begin{enumerate}
\item \label{it:for2} 
$\cl(A)\cap \Cov=\emptyset$ for every $A\in\Forb$, and
\item \label{it:aut2} 
$\cl(B)\cap \Cov\neq \emptyset$ for every $B\in \Auth$.
\end{enumerate}
\end{lemma}
\begin{proof}
Let $\Cov$ be an $(\Forb,\Auth)$-covering.
For every $A\in \Forb$ and for every $C\in\Cov$, $A\nsubseteq C$, and so 
for any $A'\in\cl(A)$ we have $A'\nsubseteq C$. Hence  $\cl(A)\cap \Cov=\emptyset$.
For every $B\in \Auth$, there exists $C\in\Cov$ with $B\subseteq C$, i.e. $C\in\cl(B)$. Hence $\cl(B)\cap \Cov\neq \emptyset$. This concludes the proof of one implication.

For any $A\subseteq P$, if $\cl(A)\cap \Cov=\emptyset$ then $A\nsubseteq C$ for every $C\in\Cov$. 
For any $B\subseteq P$, if $\cl(B)\cap \Cov\neq \emptyset$ then there exists $C\in\Cov$ with $B\subseteq C$. Hence the converse implication holds.
\end{proof}

As a direct consequence of this lemma, we have the following result.
\begin{lemma}
\label{l:incl}
Let  $\Forb,\Forb',\Auth,\Auth'\subseteq \Par(P)$ with $\Forb'\subseteq \Forb$ and $\Auth'\subseteq \Auth$. Every $(\Forb,\Auth)$-covering is also a $(\Forb',\Auth')$-covering.
\end{lemma}

The next proposition characterizes the pairs of subsets $(\Forb,\Auth)$ that admit $(\Forb,\Auth)$-coverings, and was presented in~\cite{FRR}.

\begin{proposition}
\label{p:ecov}
Let  $\Forb,\Auth\subseteq \Par(P)$. There exists an 
$(\Forb,\Auth)$-covering if and only if 
\begin{equation}\label{eq:exist}
A\nsubseteq B\text{ for every } A\in \Forb\text{ and }B\in\Auth.
\end{equation}
\end{proposition}

\begin{proof}
Let $\Cov$ be a $(\Forb,\Auth)$-covering. By Lemma~\ref{l:tcov}, for every 
$A\in \Forb$ and $B\in\Auth$, $\cl(A)\cap \Cov=\emptyset$
and $\cl(B)\cap \Cov\neq \emptyset$, so $A\nsubseteq B$.
Conversely, if $A\nsubseteq B$ for every $A\in \Forb$ and $B\in\Auth$, then $\Auth$ is an $(\Forb,\Auth)$-covering.
\end{proof}

\begin{lemma}
\label{l:incl2}
Let  $\Forb,\Forb',\Auth,\Auth'\subseteq \Par(P)$. If 
\begin{itemize}
\item for every $A'\in\Forb'$ there exists $A\in\Forb$ with $A\subseteq A'$, and 
\item for every $B'\in\Auth'$ there exists $B\in\Auth$ with $B'\subseteq B$,
\end{itemize}
then any $(\Forb,\Auth)$-covering is also a $(\Forb',\Auth')$-covering.
\end{lemma}
\begin{proof}
Let $\Cov$ be a $(\Forb,\Auth)$-covering. Let $A'\in\Forb'$ and let $A\in\Forb$ with $A\subseteq A'$. For every $C\in\Cov$ we have $A\nsubseteq C$, and so $A'\nsubseteq C$. 
Now let $B'\in\Auth'$ and let $B\in\Auth$ with $B'\subseteq B$. Then there exists a subset $C\in\Cov$ satisfying $B\subseteq C$, which also satisfies $B'\subseteq C$. Hence $\Cov$ is an $(\Forb',\Auth')$-covering.
\end{proof}

\begin{proposition}
\label{p:anti}
Let  $\Forb,\Auth,\Cov\subseteq \Par(P)$. Then $\Cov$ is an $(\Forb,\Auth)$-covering if and only if it is a
$(\min{\Forb},\max{\Auth})$-covering.
\end{proposition}
\begin{proof}
By Lemma~\ref{l:incl}, every $(\Forb,\Auth)$-covering is a $(\min{\Forb},\max{\Auth})$-covering. The converse implication is a direct consequence of Lemma~\ref{l:incl2}.
\end{proof}

According to the previous proposition, we can always restrict the search of $(\Forb,\Auth)$-coverings to the case where $\Forb$ and $\Auth$ are antichains. Further, as a consequence of Lemma~\ref{l:incl2} we can define a partial hierarchy among the pairs of antichains $(\Forb,\Auth)$. For example, every $(\{\{1,2\}\},\{\{3,4,5\}\})$-covering is also a $(\{\{1,2,3\}\},\{\{3,4\}\})$-covering.

To conclude this section, we describe a theoretical lower bound on the size of $(\Forb,\Auth)$-coverings. Note that, in the case $\Auth=\{\{i\}\, :\, i\in P\}$ and $\Forb\subseteq\binom{P}{2}$, the problem of finding an $(\Forb,\Auth)$-covering is equivalent to the graph coloring problem on the graph $G=(P,\Forb)$. In this case, the size of an optimal $(\Forb,\Auth)$-covering is just the chromatic number $\chi(G)$.

Existing general lower bounds on the chromatic number include the clique number, the minimum degree bound, Hoffman's bound, the vector chromatic number, Lov\'asz number and the fractional chromatic number. Our proposed bound generalizes to the case of $(\Forb,\Auth)$-coverings the minimum degree bound $\chi(G)\ge\frac{n}{n-\delta(G)}$, where $n$ is the number of vertices and $\delta(G)$ is the minimum degree of $G$.

\begin{proposition}\label{p:lowerbound}
Let $\Forb,\Auth\subseteq\Par(P)$ be families of subsets satisfying condition (\ref{eq:exist}), and let $\Cov$ be an $(\Forb,\Auth)$-covering. Then $$|\Cov|\ge\frac{|\Auth|}{|\Auth|-\max_{A\in\Forb}\min_{a\in A}\deg_{a}(\Auth)}.$$
\end{proposition}

\begin{proof}
Let $\Cov$ be an $(\Forb,\Auth)$-covering. Given $C\in\Cov$, denote $\Auth_{C}=\Auth\cap\Par(C)$.

By the properties of $(\Forb,\Auth)$-coverings we have that $\Auth\subseteq\cup_{C\in\Cov}\Par(C)$, and this implies that $\Auth=\cup_{C\in\Cov}\Auth_{C}$. Hence $|\Auth|\le\sum_{C\in\Cov}|\Auth_{C}|\le|\Cov|\cdot\max_{C\in\Cov}|\Auth_{C}|$, and so $|\Cov|\ge|\Auth|/\max_{C\in\Cov}|\Auth_{C}|$. We now proceed to upper bound $\max_{C\in\Cov}|\Auth_{C}|$. 

Since for every $B\in\Auth_{C}$ we have $B\subseteq C$, we see that $\cup_{B\in\Auth_{C}}B\subseteq C$. Therefore, by the definition of $(\Forb,\Auth)$-coverings we have that $A\nsubseteq\cup_{B\in\Auth_{C}}B$ for every $A\in\Forb$. Denote by $\alpha(\Forb,\Auth)$ the size of the largest subfamily of $\Auth$ with this property, i.e.
$$\alpha(\Forb,\Auth)=\max\{|\Auth'|\, :\, \Auth'\subseteq\Auth\text{ and }A\nsubseteq\cup_{B\in\Auth'}B\text{ for every }A\in\Forb\}.$$
By the preceding observation, we get that $\max_{C\in\Cov}|\Auth_{C}|\le\alpha(\Forb,\Auth)$. By definition of $\alpha(\Forb,\Auth)$, given any set $A\in\Forb$ we have $\alpha(\Forb,\Auth)\le\alpha(\{A\},\Auth)$, and so $\alpha(\Forb,\Auth)\le\min_{A\in\Forb}\alpha(\{A\},\Auth)$. Now, given a set $A\in\Forb$, a family $\Auth'\subseteq\Auth$ satisfies $A\nsubseteq\cup_{B\in\Auth'}B$ if and only if there exists an element $a\in A$ such that $a\notin\cup_{B\in\Auth'}B$. Therefore $\alpha(\{A\},\Auth)=\max_{a\in A}\alpha(\{a\},\Auth)$. Finally, by definition we see that $\alpha(\{a\},\Auth)=|\Auth|-\deg_{a}(\Auth)$. By composing the obtained results, we see that $\max_{C\in\Cov}|\Auth_{C}|\le\min_{A\in\Forb}\max_{a\in A}(|\Auth|-\deg_{a}(\Auth))=|\Auth|-\max_{A\in\Forb}\min_{a\in A}\deg_{a}(\Auth)$. The proposition follows by applying the first obtained inequality.
\end{proof}

\subsection{Multi-colorings of hypergraphs}\label{sec:coloring_map_def}

In order to construct $(\Forb,\Auth)$-coverings, we will use colorings of hypergraphs. Let $\mathcal{H}=(P,E)$ be a hypergraph.
A \emph{coloring} of $\mathcal{H}$ with $k$ colors is a mapping
$\mu:P\rightarrow \{1,\dots,k\}$ such that
for every $A\in E$ there exists $u,v\in A$ with
$\mu(u)\neq \mu(v)$.

Next we describe the connection between colorings and coverings. 
Let $\mu$ be a coloring of the hypergraph $\mathcal{H}=(P,\Forb)$ with $k$ colors. Consider the family of subsets of elements in $P$ of the same color according to $\mu$. That is, consider a family of subsets $\Cov=\{C_1,\ldots,C_k\}$ that is a partition of $P$ satisfying that $\mu(j)=i$ for every $j\in C_i$ and for every $i$. 

Now consider the pair $(\Forb,\Auth)$ with $\Auth=\{\{i\}\, :\, i\in P\}$.
Observe that $\Cov$ satisfies condition 1 in Definition~\ref{def:cov}, because if a subset $A$ is in $\Forb$, then it cannot be monochromatic. 
Since each element in $P$ has a color, condition 2 is also satisfied. In order to construct $(\Forb,\Auth)$-coverings
for other families of subsets $\Auth\subseteq\Parp$, we can use sequences of colorings. 
In order to define appropriately these constructions, we consider  multi-colorings of the hypergraph. 

For any integer $k>0$, we define a  \emph{multi-coloring} of $\mathcal{H}$ of $k$ colors as a mapping $\mu:P\rightarrow \{0,1\}^k$ with
the following property: for every $A\in E$ and for every $1\leq j\leq k$, there exists $i\in A$ for which the $j$-th coordinate of $\mu(i)$ 
is $0$, namely $\mu(i)_j = 0$. If we associate each $1\leq j \leq k$ with a different color, a multi-coloring of $\mathcal{H}$ is a mapping that maps each element in $P$ to a set of at most $k$ colors. The mapping must satisfy that for every subset in $A\in E$ and for each color, at least one element in $A$ does not have this color. A sequence of colorings of a hypergraph defines a multi-coloring. A multi-coloring defines in a natural way a family of subsets, and vice-versa. Given $\mu$, we define $\Cov=\{C_i\}_{1\leq i\leq k}\subseteq \Parp$, where $C_i$ is the subset of elements of $P$ mapped to the color~$i$.

\begin{lemma}\label{lemma:multi_coloring}
Let  $\Forb,\Auth,\Cov\subseteq \Par(P)$, with $|\Cov|=k$.
Then $\Cov$ is an $(\Forb,\Auth)$-covering if and only if 
$\Cov$ defines a multi-coloring $\mu$ of $\mathcal{H}=(P,\Forb)$ of $k$ colors with the property that for every $B\in\Auth$, there exists $1\leq j\leq k$ 
for which the $j$-coordinate of $\mu(i)$ is $1$ for every $i\in B$.
\end{lemma}

\begin{proof}
Let $\Cov=\{C_1,\ldots,C_k\}$ be an $(\Forb,\Auth)$-covering . We define a multi-coloring $\mu$ of $k$ colors as follows. For every $i\in P$ and $1\leq j \leq k$, $\mu(i)_j=1$ if and only if $i$ is in $\Cov_j$. Let $B\in\Auth$, and let $C_j$ be a subset in $\Cov$ with $B\subseteq C_j$. Then $\mu(i)_j=1$ for every $i\in B$. 

Taking into account the comments detailed above, it is straightforward to prove that the converse implication also holds.
\end{proof}

We use the connection between coverings and multi-colorings to find general constructions of coverings and upper bounds on their size. Beimel, Farr\`as, and Mintz constructed efficient secret sharing schemes for very dense graphs~\cite{BFM16}. One of the techniques developed in that work is connected to our work. 
In~\cite{FRR}, that result was described in terms of $(\Forb,\Auth)$-coverings for  $\Auth=\binom{P}{k}\s\Forb$. 
Due to Lemma~\ref{l:incl}, if $\Forb\subseteq  \binom{P}{k}$, the biggest family of subsets $\Auth\subseteq  \binom{P}{k}$ admitting a $(\Forb,\Auth)$-covering is $\Auth=\binom{P}{k}\s\Forb$. 
The next lemma states the results described above in a more general way.

\begin{lemma}
\label{l:gen1}
Let $\Forb,\Auth\subseteq\Par(P)$ be families of subsets satisfying condition (\ref{eq:exist}). Let $d$ denote the degree of $\Forb$, and suppose that sets in $\Forb$ and $\Auth$ have size at most $k$.  
Then there exists an $(\Forb,\Auth)$-covering of degree $2(2kd)^{k-1}\ln n$ and size $2(2kd)^{k}\ln n$. 
\end{lemma}

\subsection{Optimal covers}

Both the optimization problem of determining the size of an optimal $(\Forb,\Auth)$-covering and the search problem of finding an optimal $(\Forb,\Auth)$-covering are NP-hard. This is so because making $\Auth=\{\{i\}\, :\, i\in P\}$ and $\Forb\subseteq\binom{P}{2}$ transforms these problems to the corresponding graph coloring problems, and so there is a trivial reduction from the known NP-hard graph coloring problems to the $(\Forb,\Auth)$-covering problems. Next, we see NP-completeness of the decisional problem.

\begin{proposition}
The problem of deciding whether an $(\Forb,\Auth)$-covering of size $t$ exists is NP-complete.
\end{proposition}

\begin{proof}
Let $\Forb, \Auth, t$ define an instance of the problem where the answer is affirmative. Given an $(\Forb,\Auth)$-covering $\Cov$ of size $t$, a checking algorithm first verifies that $\Cov$ has size $t$, that every $B\in\Auth$ is contained in some $X\in\Cov$, and that no $A\in\Forb$ is contained in any $X\in\Cov$. The running time of this checking algorithm is at most quadratic in the size of the problem input, and thus the given problem is in NP.

Now, note that the case $\Auth=\{\{i\}\, :\, i\in P\}$ and $\Forb\subseteq\binom{P}{2}$ is equivalent to the graph coloring problem. Therefore the given problem is NP-complete.
\end{proof}



%
%
%

\section{Algebraic formulation of the problem}\label{sec:algebraic_form}

In this section we present an algebraic formulation of the combinatorial problem presented in the previous section. The purpose of this formulation is to exploit algebraic techniques to find solutions to the data splitting problem for a fixed number of fragments. 

It is not unusual that graph-coloring problems are encoded to polynomial ideals \cite{de1995grobner,de2015graph,hillar2008algebraic,levy2009survey}.
In this case, the existence of a coloring is reduced to the solvability of a related system of polynomial equations over the algebraic closure of the field.
Furthermore, the weak Hilbert's Nullstellensatz theorem allows to obtain a 
\emph{certificate} that a system of polynomial has no solutions \cite{cox1992ideals}, and, consequently, that the graph is not colorable.
The focus of this section is the use of polynomial ideals and Gr\"obner basis to provide an optimal multi-coloring $\mu$ with the property
described in Lemma \ref{lemma:multi_coloring}. 
Recall that obtaining a multi-coloring $\mu$ is equivalent to finding an $(\Forb,\Auth)$-covering. 

Let $\mathcal{H} = (P,\Forb)$ be a hypergraph and $\mu$ be a multi-coloring of $\mathcal{H} = (P,\Forb)$ of $k$ colors 
with the property that for every $B\in\Auth$, there exists $1\leq j\leq k$ 
for which the $j$-coordinate of $\mu(i)$ is $1$ for every $i\in B$.
The multi-coloring $\mu$ can be seen as assignment of $\{0,1\}$ values to a set of $kn$ variables $x_{i,j}$, where $n = |P|$, $1 \leq j \leq k$, 
$1 \leq i \leq n$ and $x_{i,j} = 1$ if and only if $\mu(i)_j = 1$. 
In other words, we assign $k$ variables $x_{i,j}$ to each vertex $i$ in $P$ in such a way that, 
\begin{displaymath}
 x_{i,j} = \begin{cases}
 1 & \text{if vertex } i \text{ takes color } j \text{ by applying } \mu\\
 0 & \text{ otherwise.}\\
 \end{cases}
\end{displaymath}
Encoding $\mu$ to a polynomial ring allows an algebraic 
formulation of the multi-coloring problem.  
Since we focus on optimal multi-colorings, the number of colors is fixed to a designated minimal $k$.
Furthermore, each variable $x_{i,j}$ takes values in $\{0,1\}$, which allows working over $\mathbb{F}_2$.

Therefore, given $X = (x_{i,j})_{1 \leq j \leq k,1 \leq i \leq n}$,  
we define the \emph{$k$-coloring ideal} $I_k(\mathcal{H},\Auth) \subset \mathbb{F}_2[X]$ 
to be the ideal generated by:
\begin{itemize}
  \item $G_1 = \{ \prod_{i \in A} x_{i,j} : 1 \leq j \leq k, A \in \Forb \}$\\
- all vertices belonging to an edge set $A \in \Forb$ cannot have the same color;
  \item $G_2 = \{ \prod_{j=1}^{k}(\prod_{i \in B} x_{i,j} -1) : B \in \Auth \}$\\
- there exists a color $j$ such that all the vertices in $B$ are colored with $j$. 
\end{itemize}

Theorem \ref{thm:GB} proves that finding a solution of $I_k(\mathcal{H},\Auth)$ is equivalent to obtaining a multi-coloring $\mu$.

\begin{theorem}\label{thm:GB}
Let $\mu:P\to \{0,1\}^k$ a multi-coloring of $\mathcal{H}=(P,\Forb)$, and assume $\cup_{B\in \Auth} B=P$. 
Then $\mu$ defines an $(\Forb,\Auth)$-covering (in the sense of Lemma \ref{lemma:multi_coloring})  
if and only if $I_k(\mathcal{H},\Auth)$ has a common root in 
$\mathbb{F}_2[X]$. In other words, the multi-coloring $\mu$ of $\mathcal{H}$ does not define an $(\Forb,\Auth)$-covering if and only if $I_k(\mathcal{H},\Auth) = (1)$.
\end{theorem}
\begin{proof}
$\mu$ is a multi-coloring map if it respects:
\begin{itemize}
 \item[P1)] for every $A\in \Forb$ and for every $1\leq j\leq k$, there exists $i\in A$ for which $\mu(i)_j$ is $0$;\label{G2}
 \item[P2)] for every $B\in\Auth$, there exists $1\leq j\leq k$ for which the $j$-coordinate of $\mu(i)$ is $1$ for every $i\in B$.\label{G3}
\end{itemize}
It is known that if a polynomial $e_1$ encodes a property and $e_2$ encodes another property, than the ideal generated by $e_1$ and $e_2$ 
encodes the conjunction (i.e., \emph{and}) of the properties. 
Therefore, if $G_1$ and $G_2$ encode the properties P1 and P2, respectively, then $I_k(\mathcal{H},\Auth)$ encodes $\mu$.
\begin{itemize}
\item $G_1$: for all $A \in \Forb$ and for every color $j$, we have that $\prod_{i \in A} x_{i,j} \in I_k(\mathcal{H},\Auth)$. This happens if and only if 
$\prod_{i \in A} x_{i,j}$ is $0$ iff there exists $i \in A$ such that  $x_{i,j}=0$, 
which is equivalent to say that there exists $i \in A$ such that  $\mu(i)_j =0$. 
\item $G_2$: for all $B \in \Auth$, we have $\prod_{j=1}^{k}(\prod_{i \in B} x_{i,j} -1) \in I_k(\mathcal{H},\Auth)$.  This happens if and only if
$\prod_{j=1}^{k}(\prod_{i \in B} x_{i,j} -1)$ is $0$ iff there exists $j$ such that $\prod_{i \in B} x_{i,j} = 1$ iff there
exists $j$ such that $x_{i,j} = 1$ for all $i$ which is equivalent to say that there exists $j$ such that $\mu(i)_j = 1$ for all $i \in B$. 
\end{itemize} 
\end{proof}

Observe that imposing $\cup_{B\in \Auth} B=P$ is not restrictive. In fact, it is always possible to add the singletons of 
any vertices to $\Auth$ in order to guarantee that a color is assigned to every vertex, without changing the request of the problem 
(see Example \ref{ex:gb} for more details). In particular, 
if $\cup_{B\in \Auth} B = P$, then for all $i \in P$ there exists $B \in \Auth$ such that $i \in B$ and therefore, there exists
a color $ 1 \leq j \leq k$ such that $\prod_{i \in B} x_{i,j} -1$ which is equivalent to say that there exists $j$ such that $x_{i,j} = 1$, which is equivalent to say that there exists $j$ 
such that $\mu(i)_j = 1$.
The hypothesis $\cup_{B\in \Auth} B = P$ has also been stated in Section \ref{sec:comb_approach}. 
%

Now that the data splitting problem is stated as an algebraic problem, a technique based on Gr\"obner basis can be used to solve it. 
Gr\"obner basis is a generating set of an ideal $I$ in a polynomial ring which allows to determine if any polynomial belongs 
to $I$ or not \cite{cox1992ideals}. In other words, it allows to determine the variety associated to $I$, i.e. the solutions of $I$.
It is proven that it is possible to associate a Gr\"obner basis to any polynomial ideal~\cite{cox1992ideals}.
Informally, Gr\"obner basis computation can be viewed as a generalization of Gaussian elimination for non-linear equations.

In our case, Gr\"obner basis can be used to find the solutions of $I_k(\mathcal{H},\Auth)$. 
Once the Gr\"obner basis of the $k$-coloring
ideal is obtained, the associated variety can be computed easily. 
The complexity of computing the Gr\"obner basis of a system of polynomial equations of degree $d$ in $n$ variables has been proven to be
$d^{O(n^2)}$ when the number of solutions is finite \cite{augere1993efficient}. In general, its complexity is $2^{2^{O(n)}}$.
Since $I_k(\mathcal{H},\Auth)$ belongs to $\mathbb{F}_2[X]$, then it has a finite number 
of solutions, and so the Gr\"obner basis complexity bound is $(kn)^{O(n^2)}$,  which represents the worst-case complexity.
The Gr\"obner basis complexity is at least that of polynomial-system solving.  

As stated before, it is possible to derive a certificate that a system of polynomials has no equation from the weak Hilbert's Nullstellensatz. In our case,
this allows to prove that it is not possible to find a multi-coloring $\mu$ with a designated number of colors $k$. 
\begin{theorem}[weak Hilbert's Nullstellensatz \cite{cox1992ideals}]\label{thm:NulLA}
Suppose that $f_1, \dots, f_m \in \mathbb{K}[x_1,\dots,x_n]$. Then there are no solutions to the system $\{f_i=0\}_{i=1,\dots,m}$ in the algebraic closure 
of $\mathbb{K}$ if and only if there exist $\alpha_1, \ldots, \alpha_m \in \mathbb{K}[x_1,\dots,x_n]$ such that
\begin{equation}\label{eq:NulLA}
\alpha_1 f_1 + \dots + \alpha_m f_m = 1  
\end{equation}
\end{theorem}
The set $\{\alpha_{i}\}$ is called a \emph{Nullstellensatz certificate}. The complexity of computing a certificate depends on the degree of $\{\alpha_i\}$, which is defined as the maximum degree of any $\alpha_i$. Fast results have been achieved for small constant degree of the Nullstellensatz certificate \cite{loera2009expressing}.

According to Theorem \ref{thm:NulLA}, using methods to compute a Nullstellensatz certificate it is possible to find out if $I_k(\mathcal{H},\Auth)$ has solutions or not.
A tentative number of colors $k$ is fixed, and then the problem is solved by applying a Nullstellensatz certificate method.
If we find that there exists no Nullstellensatz certificate, then $I_k(\mathcal{H},\Auth)$ does not have common root. 
The complexity of Nullstellensatz certificate and Gr\"obner basis methods grows with the number of variables which, in our case, 
grows with the number of colors. Therefore, it is convenient to start with few colors and increase them until a certificate of 
feasibility is found or until a Gr\"obner basis is computed.
Consider that finding the optimal $k$ is a NP-complete problem, because it has complexity equivalent to solving the system of equations.


\begin{example}\label{ex:gb}

Given $\Forb = \{\{1, 2,3\}\}$ and 
$\Auth = \{\{1, 4\},\{2, 4\},\{3\}\}$, we want to compute an $(\Forb,\Auth)$-covering.
As explained above, the problem can be encoded to polynomial ideals.
We assign $k$ variables to each attribute, where $k$ is the number of colors needed to obtain an optimal covering. 
For example, we can encode vertex $1$ to $x_{1,1}$ and $x_{1,2}$ when $2$ colors are considered. 
The variable $x_{i,j}$ is equal to $1$ if and only if vertex $i$ takes color $j$, and to $0$ otherwise.   
  
Therefore, $I_2(\mathcal{H},\Auth)$ is the ideal generated by the polynomials in $G_1$ and $G_2$, where
\begin{itemize}
 \item $G_1 = \{x_{1,1}x_{2,1}x_{3,1}, \ x_{1,2}x_{2,2}x_{3,2}\}$.
 \item $G_2 = \{x_{2,1}x_{4,1}x_{2,2}x_{4,2} + x_{2,1}x_{4,2} + x_{2,2}x_{4,1} + 1, \ x_{1,1}x_{4,1}x_{1,2}x_{4,2} + x_{1,1}x_{4,1} + 
 x_{1,2}x_{4,2} + 1, \ x_{3,1}x_{3,2} + x_{3,1} + x_{3,2} + 1\}$.
\end{itemize}
Note that there does not exist an $(\Forb,\Auth)$-covering of size one, because $\{1,2,3\}\subseteq\cup_{B\in\Auth}B$. Since we can compute the Gr\"obner basis of $I_2(\mathcal{H},\Auth)$, there exist $(\Forb,\Auth)$-coverings of size two and we obtain the optimal $(\Forb,\Auth)$-coverings:
\begin{align*}
&\{x_{1,1}, \ x_{2,1}+1, \ x_{3,1}, \ x_{4,1}+1, \ x_{1,2}+1,\ x_{2,2}, \ x_{3,2}+1, \ x_{4,2}+1\},\\
&\{x_{1,1},\ x_{2,1}+1, \ x_{3,1}+1,\ x_{4,1}+1, \ x_{1,2}+1, \ x_{2,2}, \ x_{3,2}, \ x_{4,2}+1\},\\
&\{x_{1,1}, \ x_{3,1}+1, \ x_{1,2}+1, \ x_{2,2}+1, \ x_{3,2}, \ x_{4,2}+1\}.
\end{align*}
In the last solution, the variables $x_{2,1}$ and $x_{4,1}$ are missing, which means that they can take both $0$ and $1$ values.
Therefore, the solutions can be re-written as the following coverings:  
\begin{displaymath}
\begin{array}{llll}
\{\{2,4\},\{1,3,4\}\},&\{\{2,3,4\},\{1,4\}\},&\{\{2,3\},\{1,2,4\}\},&\{\{2,3\},\{1,2,4\}\},\\ 
\{\{3,4\},\{1,2,4\}\},&\{\{2,3,4\},\{1,2,4\}\},&\{\{3\},\{1,2,4\}\}.\\
\end{array}
\end{displaymath}

To obtain the number of colors $k$ which allows to compute an optimal covering, a tentative $k$ is fixed starting by $k=2$. 
If the ideal $(1)$ is obtained as the result of the Gr\"obner basis method applied to $I_2(\mathcal{H},\Auth)$,
then the next $k$ is considered until the solution is different from $(1)$. This $k$ is the smallest one for which we have an $(\Forb,\Auth)$-covering, and thus it is optimal.
\end{example}

\section{A greedy algorithm}\label{section:effalgorithm}

In this section we aim for an efficient method to build $(\Forb,\Auth)$-coverings and for upper bounds on the size of an optimal $(\Forb,\Auth)$-covering. 

As seen above, the problem of finding an optimal $(\Forb,\Auth)$-covering is NP-hard. Hence, as expected, the labour involved in finding an optimal $(\Forb,\Auth)$-covering can render methods inefficient when solving practical data splitting instances. Our strategy to circumvent this consists in sacrificing optimality to achieve a polynomial-time algorithm.

The problem of finding upper bounds on the size of an optimal $(\Forb,\Auth)$-covering $\Cov$ has been studied in the literature for the following particular cases:
\begin{itemize}
\item In the case $\Auth=\{\{i\}\, :\, i\in P\}$ and $\Forb\subseteq\binom{P}{2}$, the problem of finding an $(\Forb,\Auth)$-covering is easily seen to be equivalent to the graph coloring problem. Then $|\Cov|$ is the chromatic number of the graph $G=(\Auth,\Forb)$. For instance, the greedy coloring bound gives $|\Cov|\le \deg(\Forb)+1$. 

\item The case $\Auth\subseteq\binom{P}{2}\cup P$ and $\Forb=\binom{P}{2}\backslash\Auth$ has been studied as the clique covering and the clique partition numbers. Hall \cite{H41} and Erd\H{o}s et al \cite{EGP66} showed that $|\Cov|\le\lfloor |P|^{2}/4\rfloor$. 

\item In the case $\Auth=\binom{P}{l}$ and $\Forb\subseteq\cup_{i>k}\binom{P}{i}$ for $1\le l\le k<n$, the problem of finding a $k$-uniform $(\Forb,\Auth)$-covering is equivalent to finding an $(n,k,l)$-covering design. In this case, Spencer \cite{S87} showed that $|\Cov|\le\binom{n}{l}\left/\binom{k}{l}\left( 1+\ln\binom{k}{l} \right)\right.$.
\end{itemize}

In the following, we first describe a general upper bound on the size of an optimal $(\Forb,\Auth)$-covering. We then deduce from this bound an algorithm to build $(\Forb,\Auth)$-coverings, and analyze its worst-time complexity. Finally, we introduce an heuristic improvement and a theoretical bound that improve the prior results for sparse enough $\Auth$.

\subsection{Our construction}

The next result generalizes the greedy coloring bound to $(\Forb,\Auth)$-coverings. It gives a general bound of the size of an optimal $(\Forb,\Auth)$-covering in terms of the degrees of $\Forb$ and $\Auth$.

\begin{theorem}\label{prop:greedy}
Let $\Forb,\Auth\subseteq\Par(P)$ be families of subsets satisfying condition (\ref{eq:exist}), and suppose that sets in $\Auth$ have size at most $k$.  
Then there exists an $(\Forb,\Auth)$-covering $\Cov$ of size $$|\Cov|\le k\deg(\Forb)\deg(\Auth)+1$$ such that $\deg_{v}(\Cov)\le \deg_{v}(\Auth)$ for every $v\in P$. 
\end{theorem}

\begin{proof}
We prove this by induction on $|\Auth|$.
If $|\Auth|=1$, then $\Cov=\Auth$ satisfies the lemma.
Now let $s>1$ be an integer and assume that the proposition holds for every pair $\Forb,\Auth\subseteq\Par(P)$ of families of subsets satisfying $|\Auth|<s$ and the proposition hypotheses. Let $\Forb,\Auth'\subseteq\Par(P)$ be a pair of families of subsets satisfying $|\Auth'|=s$ and the proposition hypotheses, and express $\Auth'=\Auth\cup\{B\}$ for some fixed $B\in\Auth'$ and $|\Auth|=s$. Then, by induction hypothesis, there exists an $(\Forb,\Auth)$-covering $\Cov$ with $|\Cov|\le k\deg(\Forb)\deg(\Auth)+1$ and such that every $v\in P$ is contained in at most $\deg_{v}(\Auth)$ elements of $\Cov$. We now build an $(\Forb,\Auth')$-covering $\Cov'$ from $\Cov$, in such a way that $|\Cov'|\le k\deg(\Forb)\deg(\Auth')+1$ and that every $v\in P$ is contained in at most $\deg_{v}(\Auth')$ elements of $\Cov$.

If $B$ is contained in some $X\in\Cov$, then $\Cov'=\Cov$ satisfies the lemma. Otherwise, let
$$ \mathcal{F}_{B}=\{X\in\Cov : \text{there exists }A\in\Forb\text{ with }A\subseteq X\cup B\}. $$
Note that the condition $A\subseteq X\cup B$ is equivalent to $A\cap B\neq\emptyset$ and $A\setminus B\subseteq X$. Since there are at most $k\deg(\Forb)$ elements $A\in \Forb$ with $A\cap B\neq\emptyset$, and since every set of  the form $A\setminus B$ can be contained in at most $\deg(\Auth)$ elements of $\Cov$ (because $\deg_{v}(\Cov)\le \deg_{v}(\Auth)$ for every $v\in P$ by hypothesis), we have that $|\mathcal{F}_{B}|\le k\deg(\Forb)\deg(\Auth)$.

Therefore, either there exists an element $X\in\Cov\setminus\mathcal{F}_{B}$, in which case we take
$$\Cov'=\{X\cup B\}\cup(\Cov\setminus\{X\})$$
or $\Cov=\mathcal{F}_{B}$, in which case $|\Cov|\le k\deg(\Forb)\deg(\Auth)$ and we let $ \Cov'=\Cov\cup\{B\}. $
\end{proof}

 Algorithm \ref{greedyalg} is a greedy algorithm to compute an $(\Forb,\Auth)$-covering that follows directly from the constructive proof of the previous lemma. This algorithm simply builds a ordered $(\Forb,\Auth)$-covering $\Cov$ by iterating through $\Auth$. Every set $B$ in $\Auth$ is merged with the first available element of $\Cov$, i.e., with the first element $X\in\Cov$ such that no $A\in\Forb$ is contained in $X\cup B$. If no such $X$ exists, then $B$ is added as a singleton in $\Cov$. 
Note that this algorithm is a generalization of the usual greedy coloring algorithm.

 \RestyleAlgo{boxruled}
\begin{algorithm} \LinesNumbered
  \caption{Construction}\label{greedyalg}
  \KwIn{$\Forb=\{A_{1},\ldots,A_{r}\}$, $\Auth=\{B_{1},\ldots,B_{s}\}$}
  Initialize $\Cov\leftarrow\emptyset$\
  \For{$i=1,\ldots,s$}
  {
  	\If{$B_{i}$ is not contained in any $X\in\Cov$}
  	{
  		\eIf{there exists $X\in\Cov$ such that $A\not\subseteq X\cup B_{i}$ for every $A\in\Forb$}
  		{
  			$\Cov\leftarrow\{X\cup B_{i}\}\cup(\Cov\setminus\{X\})$
  		}{
  			$\Cov\leftarrow\Cov\cup\{B_{i}\} $
  		}
  	}
  }
  \KwOut{The $(\Forb,\Auth)$-covering $\Cov$}
\end{algorithm}

To see the worst-time complexity of Algorithm \ref{greedyalg}, note that the first loop (line 2) is repeated $|\Auth|$ times. At step $i$, the first \verb|if| statement (line 3) requires checking at most $i-1$ inclusions, and the second \verb|if| statement (line 4) requires checking at most $(i-1)·|\Forb|$ inclusions. Therefore, Algorithm \ref{greedyalg} runs in time $O(|\Forb|\cdot|\Auth|^{2})$.

\subsection{An heuristic improvement}\label{sec:heur}

In order to motivate the heuristic procedure proposed later, we must first note that the output of Algorithm \ref{greedyalg} depends strongly on the particular order in which elements of $\Auth$ are taken in the first loop. In particular, we see in the following proposition that there always exists an optimal ordering of the elements of $\Auth$. Of course, since the problem of finding an optimal $(\Forb,\Auth)$-covering is NP-hard and an optimal ordering can be verified in polynomial time, finding an optimal ordering in our case is NP-complete.

\begin{proposition}
Let $\Forb,\Auth\subseteq\Par(P)$ be families of subsets satisfying condition (\ref{eq:exist}). Then there exists an ordering of $\Auth$ such that Algorithm \ref{greedyalg} outputs an optimal $(\Forb,\Auth)$-covering.
\end{proposition}

\begin{proof}
Let $\Cov=\{X_{1},\ldots,X_{t}\}$ be an optimal $(\Forb,\Auth)$-covering. For every $j\in\{1,\ldots,t\}$, define $S_{j}$ to be the family of elements of $\Auth$ that are contained in $X_{j}$ and that are not contained in any $X_{l}$ for $l<j$,
$$ S_{j}=\{B\in\Auth : B\subseteq X_{j}\text{ and } B\not\subseteq X_{l}\text{ for every } l<j\}.$$
We first prove that $\{S_{j}\}_{j=1}^{t}$ defines a partition of $\Auth$. 

Indeed, $S_{j}\neq\emptyset$, because otherwise $\Cov\setminus\{X_{j}\}$ would be an $(\Forb,\Auth)$-covering smaller than $\Cov$. 

Also, $S_{i}\cap S_{j}=\emptyset$ for every $i,j$. Otherwise, if $B\in S_{i}\cap S_{j}$ with $i<j$, then $B\in S_{i}$ implies $B\subseteq X_{i}$, and $B\in S_{j}$ implies $B\not\subseteq X_{i}$, a contradiction. 

Finally, since every $B\in\Auth$ is contained in some element of $\Cov$ by the definition of $(\Forb,\Auth)$-covering, we can take $X_{j}\in\Cov$ with minimal index $j$ among those that contain $B$. Then $B\in S_{j}$ by definition, and therefore $\Auth=\cup_{j=1}^{t} S_{i}$.

Now, define a new ordering of $\Auth$ by taking the sets in $S_{j}$ sequentially. That is, if $S_{j}=\{B_{j,1},\ldots,B_{j,k_{j}}\}$, define
$$ \Auth'=\{B_{1,1},\ldots,B_{1,k_{1}},\ldots,B_{t,1},\ldots,B_{t,k_{t}}\}. $$

Consider the behavior of Algorithm~\ref{greedyalg} on input $\Forb,\Auth'$. It is easy to see that, when the algorithm finishes processing the sets in $S_{j}$, the local variable $\Cov$ holds at most $j$ elements. Therefore, since the covering $\Cov$ is optimal, Algorithm \ref{greedyalg} outputs an optimal $(\Forb,\Auth)$-covering.
\end{proof}

Following this result, we propose an heuristic procedure to build an ordering of $\Auth$, inspired in the Welsh-Powell algorithm \cite{WP}. This procedure can be deduced from the proof of the following proposition, which effectively reduces the upper bound given in Theorem~\ref{prop:greedy} for sparse enough $\Auth$. 

\begin{proposition}\label{prop:heuristic}
Assume the hypotheses of Theorem~\ref{prop:greedy}. Then there exists an $(\Forb,\Auth)$-covering $\Cov$ of size $$|\Cov|\le \max_{i}\min\{\deg_{B_{i}}(\Forb)\deg(\Auth)+1,i\}$$ such that $\deg_{v}(\Cov)\le \deg_{v}(\Auth)$ for every $v\in P$. 
\end{proposition}

\begin{proof}
First reorder $\Auth$ so that $\Auth=\{B_{1},\ldots,B_{s}\}$ satisfies
$$ \deg_{B_{1}}(\Forb)\ge \deg_{B_{2}}(\Forb) \ge \cdots \ge \deg_{B_{s}}(\Forb). $$

Now, consider the behavior of Algorithm~\ref{greedyalg} on input $\Forb$ and the reordered $\Auth$. At step $i$, algorithm \ref{greedyalg} processes $B_{i}$. In this step, there can be at most $\deg_{B_{i}}(\Forb)\deg(\Auth)$ sets $X\in\Cov$ such that $B_{i}$ does not satisfy the condition in line $4$ (that is, such that there exists $A\in\Forb$ with $A\subseteq X\cup B_{i}$). To see this, note that by definition at most $\deg_{B_{i}}(\Forb)$ elements $A\in\Forb$ intersect $B_{i}$, and that each set of the form $A\backslash B_{i}$ can be contained in at most $\min\{\deg(\Auth),|\Cov|\}\le \deg(\Auth)$ elements of $\Cov$. 

Now, at step $i$ the number of elements $X\in\Cov$ checked in the condition of line $4$ is at most $\min\{\deg_{B_{i}}(\Forb)\deg(\Auth),|\Cov|\}$. 
Since at step $i$ the family $\Cov$ has at most $i-1$ sets, at most $\min\{\deg_{B_{i}}(\Forb)\deg(\Auth),i-1\}$ elements of $\Cov$ are checked until either line $5$ or $7$ is executed, and line $7$ can add an additional element to $\Cov$. Hence, by iterating through all elements of $\Auth$, the size of the final output can be at most $\max_{i}\min\{\deg_{B_{i}}(\Forb)\deg(\Auth),i-1\}+1$.
\end{proof}

We now state our heuristic improvement of Algorithm~\ref{greedyalg}, which follows directly from the previous proof.

 \RestyleAlgo{boxruled}
\begin{algorithm} \LinesNumbered
  \caption{Heuristic Improvement of Algorithm~\ref{greedyalg}}\label{heuralg}
  \KwIn{$\Forb=\{A_{1},\ldots,A_{r}\}$, $\Auth=\{B_{1},\ldots,B_{s}\}$}
  \For{$B\in\Auth$}
  {
    Compute $ \deg_{B}(\Forb)=|\{A\in\Forb: A\cap B\neq\emptyset\}| $
  }
  Sort $\Auth$ so that $\Auth=\{B'_{1},\ldots,B'_{s}\}$ satisfies
  $ \deg_{B'_{1}}(\Forb)\ge \deg_{B'_{2}}(\Forb) \ge \cdots \ge \deg_{B'_{s}}(\Forb)$

  \KwOut{The output of Algorithm~\ref{greedyalg} on input $\Forb,\Auth$}
\end{algorithm}

To see the worst-time complexity of Algorithm \ref{heuralg}, note that the computation of each quantity $\deg_{B}(\Forb)$ requires $O(|\Forb|)$ time. Adding in the sorting time, we see that our heuristic takes $O(|\Auth|\cdot(|\Forb|+\log(|\Auth|)))$ time. In turn, adding this to the cost of Algorithm~\ref{greedyalg} does not alter the total $O(|\Forb|\cdot|\Auth|^{2})$ worst-time complexity.

\begin{remark}
In fact, the previous proof indicates a slightly better bound. For an integer $\alpha$ define the function $f_{\alpha}$ by $$f_{\alpha}(\beta)=\left\{\begin{array}{ll} \beta & \text{if }\alpha<\beta \\ \beta+1 & \text{otherwise.} \end{array}\right.$$ Then the bound on the previous proposition can be replaced by $$|\Cov|\le f_{\deg_{B_{s}}(\Forb)\deg(\Auth)}(f_{\deg_{B_{s-1}}(\Forb)\deg(\Auth)}(\cdots (f_{\deg_{B_{2}}(\Forb)\deg(\Auth)}(1))\cdots))+1.$$
\end{remark}

\section{Experimental results}\label{sec:exp_results}

This section details the experimental results obtained by implementing the
proposed methods in the Sage Mathematical Software System \cite{sage}, version 7.4.  First, a comparison between Algorithm \ref{heuralg} and the Gr\"obner 
basis method on a practical setting is shown. 
Then, a performance analysis of Algorithms \ref{greedyalg} and \ref{heuralg} 
is carried out over random graphs.
The reported experiments have been conducted on an AMD Ryzen 7 1700X Eight-core 3.4 GHz processor, with 32 GB of RAM, in Sage~\cite{sage} and under Ubuntu 4.10.0-37. All experiments have been carried out without parallelization. All CPU running times have been collected using the function \verb|cputime(subprocesses=True)| in Sage. 

%

\subsection{Medical data example}\label{sec:medical_comparison}

Medical data tend to be extremely storage-demanding, and thus it is often unfeasible for the data holder to store it in local.
Therefore, medical data provides a good candidate to apply privacy-preserving data splitting.

Table~\ref{tab:hospital} depicts several possible features that can be found in medical data.
Since the features \emph{patient ID} and \emph{address} completely identify the patient, 
they need to be stored in encrypted form, and therefore they are not taken into account in the associated data splitting problem.  
A different numerical identifier is assigned to every other feature for presentation purposes.

\begin{table}[ht]
\centering
 \begin{tabular}{|c|c|c|c|c|}
\cline{1-2}\cline{4-5}
 $\#$ & Hospital Folder features& \hspace{0.5cm} & $\#$ & Hospital Folder features \\
\cline{1-2}\cline{4-5}
\cline{1-2}\cline{4-5}
- & patient ID & & 4 & weight \\
- & address & & 5 & diagnosis\\
0 & ZIP code & & 6 & procedure\\
1 & birth date & & 7 & medication\\
2 & gender & & 8 & charges \\
3 & ethnicity & & 9 & hospital ID \\
\cline{1-2}\cline{4-5}
 \end{tabular}
 \caption{Example of patient Healthcare features.}\label{tab:hospital}
\end{table}

Observe that other combinations of attributes can also be sensitive. An example of such combination can be 
$\{0, 2, 3\}$ as it is shown in \cite{sweeney2000simple}, 
where a 1990 federal census reports that in Dekalb, Illinois there were only two black women who resided in that town. 
We can also consider 
$\{0,1,3\}$, $\{0,1,4\}$ and $\{1,2,3\}$ sensitive for obvious reasons. 
Moreover, some attributes need to be stored in the same fragment, for instance to perform statistical analysis computations. 
Possible combinations are: $\{1,2,5\}$, $\{1,3,5\}$ and $\{0,2,5\}$.

The Gr\"obner basis method (implemented as \verb|buchberger2()| in Sage) and Algorithm \ref{heuralg} can be used to find an $(\Forb,\Auth)$-covering 
that solves the data splitting problem, 
where $\Forb = \{\{0,2,3\},\{0, 1, 2\},\{0,1,4\},\{1,2,3\}\}$ and $\Auth = \{\{1,2,5\},\{1,3,5\},\{0,2,5\},\{4\}\}$.

\begin{table}[ht]
\centering \small
 \begin{tabular}{|c|c|c||c|c|c||c|c|c|}
\hline
\multicolumn{3}{|c||}{} & \multicolumn{3}{|c||}{Gr\"obner basis method} & \multicolumn{3}{|c|}{Algorithm \ref{heuralg}}\\
\hline
\# v. & $\Forb$ & $\Auth$ & \# sol. & $|\Cov|$ & time &  opt. & $|\Cov|$ & time \\
\hline
\hline
\multirow{2}{*}{6} & $ \{023,012,014,$ & \multirow{2}{*}{$ \{125,135,025,4\}$} 
& \multirow{2}{*}{15} & \multirow{2}{*}{3} & \multirow{2}{*}{6.68s} & \multirow{2}{*}{Yes} 
& \multirow{2}{*}{3} & \multirow{2}{*}{3.84ms} \\
& $123\}$ & & & & & &  &  \\
\hline

6 & $\{023,012,014\}$ & $\{125,135,025,4\}$ & 3 & 2 & 1.12ms & Yes & 2 & 1.19ms \\
\hline
 8 & $\{ 045,123,89\}$ & $\{124,458,09,238\}$ & 1 & 2 & 316ms & No & 3 & 1.41ms \\

\hline
 \multirow{2}{*}{9} & $\{13,168,34,$ & $\{023, 012, 36, $ &
 \multirow{2}{*}{204} & \multirow{2}{*}{3} & \multirow{2}{*}{16h 45min} & \multirow{2}{*}{Yes} & \multirow{2}{*}{3} & 
 \multirow{2}{*}{4.88ms} \\
  & $79,036\}$ & $46,78, 07, 9\}$ & & & & & & \\
\hline
 \multirow{2}{*}{10} & $\{02,168,34,$ & $\{01,128,35,46,$ &
 \multirow{2}{*}{2} & \multirow{2}{*}{2} & \multirow{2}{*}{1.87s} & \multirow{2}{*}{Yes}
 & \multirow{2}{*}{2} & \multirow{2}{*}{2.19ms} \\
  & $79,03\}$ & $78,04,23, 9\}$ & & & & & & \\

 \hline
\end{tabular}
\vspace{5pt}

\begin{minipage}{0.8\linewidth}
\raggedright{ ``\# v.'': number of vertices of the selected hypergraph,\\ 
  ``\# sol.'': number of optimal coverings, \\
  $|\Cov|$: size of the $(\Forb,\Auth)$-coverings computed by the respective method,\\ 
 ``time'': the time needed by the related method to compute the solution, \\
 ``opt.'': whether or not the solution of Algorithm \ref{heuralg} is optimal,\\  
 We use compact notation for sets, i.e. $023 = \{0,2,3\}$.}
\end{minipage}
 \caption{Comparison between Gr\"obner basis method and Algorithm \ref{heuralg} on several hypergraphs.}\label{tab:comparisonGB_Jordi}
\end{table}

The first column of Table \ref{tab:comparisonGB_Jordi} shows the results of applying the Gr\"obner basis method and Algorithm \ref{heuralg}
to the medical data set.  
Algorithm \ref{heuralg} has been chosen for the tests above instead of Algorithm \ref{greedyalg} due to efficiency reasons. 
Both the Gr\"obner basis method and Algorithm \ref{heuralg} provide optimal solutions in the considered case 
but with a considerable time difference.
Two of the optimal solutions computed by the Gr\"obner basis method are depicted in Figure \ref{fig:sol_medical}.

Table \ref{tab:comparisonGB_Jordi} also depicts the results of the Gr\"obner basis method and Algorithm \ref{heuralg} to several 
other splitting problems, all referred to the medical data set of Table \ref{tab:hospital}.
The time needed %
to obtain a solution is strictly related to the 
degree of $\Forb$ and $\Auth$. Other parameters which affect the running time are the number of vertices and the size of the optimal covering 
(see Section \ref{sec:algebraic_form} for more details). However, while having the same number of vertices, the needed time may vary greatly. Observe that Algorithm \ref{heuralg} does not always compute an optimal solution.

\begin{figure}[ht]
	\begin{centering}
		\begin{tikzpicture}[scale=0.4,thick,framed]
        \node at (-0.3,-0.8) {\LARGE 0};
        \node at (4,-0.3) {\LARGE 1};
        \node at (1.8,-2.2) {\LARGE 5};
        \node at (0.1,-3.5) {\LARGE  2};
        \node at (5.5,-2.5) {\LARGE  3};
        \node at (2.5,-4.7) {\LARGE 4};
        \draw[rotate=40]  (-1.1,-2.7) ellipse (2 and 3.5);
        \draw[rotate=70]  (-0.1,-3) ellipse (2.3 and 4.2);
        \draw[rotate=130]  (-2.8,-0.3) ellipse (1 and 4);
        \draw[draw=none] (-2.7,2.3) rectangle (7.7,-6.4);
        \end{tikzpicture}
		~~~
		\begin{tikzpicture}[scale=0.4,thick,framed]
        \node at (-0.5,-0.6) {\LARGE 0};
        \node at (4,-0.1) {\LARGE 1};
        \node at (2,-1.8) {\LARGE 5};
        \node at (0.2,-3.4) {\LARGE  2};
        \node at (5.7,-2.5) {\LARGE  3};
        \node at (2.5,-4.4) {\LARGE 4};
        \draw[rotate=0]  (0.4,-1.9) ellipse (2.5 and 2.5);
        \draw[rotate=60]  (0.6,-3.8) ellipse (2.2 and 3.2);
        \draw[rotate=145]  (-3.25,0.6) ellipse (2.3 and 4.2);
        \draw[draw=none] (-2.9,2) rectangle (7.5,-6.7);
        \end{tikzpicture}
		\par
	\end{centering}
	\caption{Two of the optimal solutions computed by Gr\"obner basis method of the medical data problem. 
	The solution found by Algorithm
	\ref{heuralg} is depicted in the right chart.
	Vertices in the same set belong to the same fragment.}\label{fig:sol_medical}
\end{figure}
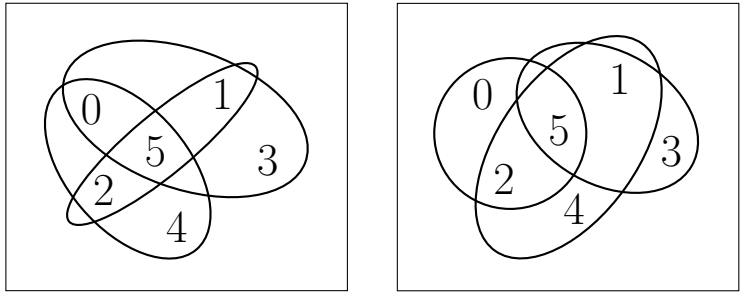

\subsection{Performance analysis over random graphs}\label{sec:random_graph_comparison}

We now give some performance measures to analyze the results presented in Section~\ref{section:effalgorithm}. In this performance analysis we restrict to the graph case, and thus we take $\Forb\subseteq\binom{P}{2}$ and $\Auth\subseteq\binom{P}{2}\cup P$ to be disjoint subsets. We classify the test cases according to two parameters: the number $n$ of vertices and the sum of densities $\rho=\rho_\Forb + \rho_\Auth$ of $\Forb$ and $\Auth$. For each test case, we randomly generate graphs by choosing every single edge of the complete graph $K_{n}$ with probability $\rho$, and we then throw a uniform random coin for each chosen edge to determine if it belongs either to $\Forb$ or to $\Auth$. Next, we add the necessary singletons to $\Auth$ so that $\cup_{B\in\Auth}B=P$. Finally, we randomly shuffle $\Forb$ and $\Auth$ and we apply the algorithm to test. Hence, in the considered experiments both $\Forb$ and $\Auth$ are generated with density $\rho_{\Forb}=\rho_{\Auth}=\rho/2$.

In Table~\ref{table:time} we analyze the time performance of Algorithm~\ref{greedyalg} and Algorithm~\ref{heuralg}. For each of the considered $n$ and $\rho$, the reported CPU running times have been averaged over $10^3$ independent random experiments.

\begin{table}[ht]\centering\begin{tabular}{|c|c|c|c|c||c|c|c|c|}
\hline
\multirow{1}{*}{} &
      \multicolumn{4}{c||}{Algorithm~\ref{greedyalg}} &
      \multicolumn{4}{c|}{Algorithm~\ref{heuralg}} \\\hline
 $n$ & $\rho=0.1$ & $\rho=0.4$ & $\rho=0.7$ & $\rho=1.0$ & $\rho=0.1$ & $\rho=0.4$ & $\rho=0.7$ & $\rho=1.0$ \\\hline
10 & 0.003 & 0.010 & 0.029 & 0.066 & 0.005 & 0.020 & 0.057 & 0.125\\\hline
20 & 0.018 & 0.181 & 0.740 & 1.989 & 0.036 & 0.347 & 1.255 & 3.021\\\hline
30 & 0.065 & 1.150 & 5.134 & 14.80 & 0.136 & 2.023 & 7.769 & 19.88\\\hline
40 & 0.186 & 4.415 & 21.10 & 61.85 & 0.384 & 7.116 & 29.03 & 78.39\\\hline
50 & 0.442 & 12.69 & 63.28 & 190.6 & 0.891 & 19.29 & 81.76 & 227.8\\\hline
60 & 0.941 & 30.45 & 155.8 & 478.7 & 1.864 & 43.95 & 193.2 & 552.9\\\hline
70 & 1.803 & 63.69 & 338.6 & 1054 & 3.452 & 88.57 & 403.8 & 1177\\\hline
\end{tabular}
\vspace{5pt}
\caption{Time performance analysis (in seconds).}\label{table:time} 
\end{table}

Observe that average running times increase both in the number of attributes and the density, and range between milliseconds and 20 minutes for the considered parameters.

Next, in Table~\ref{table:sizeheur} we compile evidence on the size of the result output by Algorithm~\ref{heuralg} over the size of the result output by Algorithm~\ref{greedyalg}. For every considered $n$ and $\rho$, we randomly instantiate $10^5$ different $\Forb$ and $\Auth$ as stated above, and for each of them we compute the sizes $s_{\mathrm{alg}}$ and $s_{\mathrm{heur}}$ of the $(\Forb,\Auth)$-coverings given by Algorithm~\ref{greedyalg} and by Algorithm~\ref{heuralg}, respectively. Then, we compute the decrease as the percentage $(100(s_{\mathrm{alg}}-s_{\mathrm{heur}})/s_{\mathrm{alg}})\%$ in size offered by the heuristic. The reported percentual decreases have been averaged over at least $10^5$ independent random experiments.

\begin{table}[ht]\centering\begin{tabular}{|c|c|c|c|c|c|c|c|c|c|c|}
\hline
 $n$ & $\rho=0.1$ & $\rho=0.3$  & $\rho=0.5$ & $\rho=0.7$ & $\rho=0.9$ & $\rho=1.0$ \\\hline
5 & 0.01233 & 0.2569 & 0.5101 & 0.3551 & 0.09567 & 0\\\hline
6 & 0.07317 & 1.006 & 1.234 & 0.7517 & 0.3193 & 0.08730\\\hline
7 & 0.2227 & 1.991 & 1.876 & 1.193 & 0.6438 & 0.3433\\\hline
8 & 0.5131 & 3.000 & 2.486 & 1.709 & 1.075 & 0.7280\\\hline
9 & 0.9214 & 3.675 & 3.007 & 2.269 & 1.600 & 1.184\\\hline
10 & 1.609 & 4.275 & 3.563 & 2.799 & 2.087 & 1.634\\\hline
\end{tabular}
\vspace{5pt}
\caption{Average percent size reduction given by Algorithm~\ref{heuralg} from the size given by Algorithm~\ref{greedyalg}.}\label{table:sizeheur}
\end{table}

Following the same procedure, in Table~\ref{table:sizeheuropt} we compile evidence on the size increase of the covering given by Algorithm~\ref{heuralg} in relation to the size of an optimal covering. The reported percentage decreases have been averaged over at least $10^4$ independent random experiments.

\begin{table}[ht]\centering\begin{tabular}{|c|c|c|c|c|c|c|c|c|c|c|}
\hline
 $n$ & $\rho=0.1$ & $\rho=0.3$ & $\rho=0.5$ & $\rho=0.7$ & $\rho=0.9$ & $\rho=1.0$ \\\hline
5 & 0.00003333 & 0.0002666 & 0.003666 & 0.001333 & 0.00001333 & 0\\ \hline
6 & 0.001333 & 0.0250 & 0.1045 & 0.07500 & 0.03203 & 0.02680\\\hline
7 & 0.009333 & 0.3333 & 0.3250 & 0.1683 & 0.1602 & 0.1548\\\hline
8 & 0.06667 & 0.5167 & 0.4083 & 0.4844 & 0.3111 & 0.1940\\\hline
9 & 0.1333 & 0.8417 & 1.041 & 0.8906 & 0.8411 & 0.3667\\\hline
10 & 0.3333 & 1.493 & 1.601 & 1.396 & 1.210 & 0.9015\\\hline
\end{tabular}
\vspace{5pt}
\caption{Average percent size increase given by Algorithm~\ref{heuralg} with respect to the optimal size.}\label{table:sizeheuropt}
\end{table}

In Table~\ref{table:sizeheur}, we observe that our heuristic algorithm~\ref{heuralg} generally improves the greedy algorithm~\ref{greedyalg} for a small number of attributes, and that this improvement grows in the number of attributes and is larger for medium densities. In addition, in Table~\ref{table:sizeheuropt} we confirm that our heuristic algorithm generally provides near-to-optimal sized decompositions for a small number of attributes, and that much better results are achieved for small densities. In the case $n=5$ and $\rho=1$, our algorithms always provide optimal coverings.

\section{Conclusion}\label{sec:conclusions}

Recent data splitting research focuses in preserving the privacy of a sensitive data set by decomposing it into a small number of fragments.
In this context, data is split into a small number of fragments, frequently two or three. Since this does not usually suffice to ensure privacy,
existing solutions build cryptographic techniques on top of data splitting. However, up to this point no research has engaged with the data 
splitting problem in a setting where no other privacy-preserving techniques are to be used.

In this paper, we tackle the problem of addressing privacy concerns by finding a decomposition into fragments of a given data set. 
We also take into account processing constraints, which may impose some sets of data attributes to be stored together in the same fragment.
We first consider the problem of finding the optimal number of fragments needed to satisfy privacy and processing constraints, and we 
further remove the optimality condition to provide better efficiency.

Firstly, we present a formulation of the stated problem and a concrete approach to solve it. The data splitting problem is presented as 
a purely combinatorial problem, by specifying two families of subsets $\Forb$ and $\Auth$. The first family $\Forb$ represents privacy constraints,
and specifies sets of attributes that must not be stored together for privacy reasons. The second family $\Auth$ represents processing 
constraints, and specifies sets of attributes that must be stored together in the same fragment in order to speed up processing. 
In this setting, we introduce the notion of $(\Forb,\Auth)$-covering, and show that $(\Forb,\Auth)$-coverings directly translate to solutions of the 
privacy-preserving data splitting problem.

Once the combinatorial problem of finding $(\Forb,\Auth)$-covering is stated, we show that it can be solved by using purely algebraic techniques 
through its equivalence to a hypergraph-coloring problem. We thus exhibit an algebraic formulation of the data splitting problem, 
which translates privacy and processing constraints to a system of simultaneous equations. Through the use of Gr\"obner basis, this formulation 
allows the computation of optimally-sized data decompositions.

Since finding an optimal covering is an NP-hard problem, obtaining optimal solutions is often unfeasible in practice. We hence present a 
new greedy algorithm that sacrifices optimality for efficiency, achieving a polynomial running time in the size of the considered problem. 
We further present an heuristic improvement of this greedy algorithm, that provides smaller decompositions when the family of constraints is 
sparse enough.

A performance analysis is carried out to evaluate all of the presented solutions. First, we analyze the execution time of our first algebraic
approach in the context of a medical data set. Next, we report the execution times of our greedy and heuristic algorithms over random graphs, 
and we estimate the size overhead incurred by both algorithms with respect to the optimal size for a small number of attributes. The experimental results show that our greedy 
algorithm requires milliseconds to find a solution, whereas computing an optimal solution may require hours depending on the problem at hand.

\section*{Acknowledgments}

The present work was supported by the European Comission through H2020-ICT-2014-1-644024 ``CLARUS" and H2020-DS-2015-1-700540 ``CANVAS" and by the Government of Spain through TIN2014-57364-C2-1-R ``SmartGlacis" and TIN2016-80250-R ``Sec-MCloud".


\begin{thebibliography}{11}


\bibitem{abu2010racs}
Abu-Libdeh, H., Princehouse, L., Weatherspoon, H.: 
RACS: a case for cloud storage diversity. 
\emph{Proceedings of the 1st ACM symposium on Cloud computing}. ACM, 2010.

\bibitem{GM}
Aggarwal, G., Bawa, M., Ganesan, P., Garcia-Molina, H., Kenthapadi, K., Motwani, R., Srivastava, U., Thomas, D., Xu, Y.:
Two Can Keep a Secret: A Distributed Architecture for Secure Database Services.
\emph{Conference on Innovative Data Systems Research 2005}. pp. 186-199, 2005.


\bibitem{BFM16}
Beimel, A., Farr\`as, O., Mintz, Y.:
Secret Sharing Schemes for Very Dense Graphs.
{\em J. of Cryptology}, 29(2): 336--362, 2016.

\bibitem{B} Br\'elaz D.: 
New methods to color the vertices of a graph. 
\emph{Communnications of the ACM} 22, 4, 251-256, 1979.

\bibitem{brinkman2006lucky}
Brinkman, R., Maubach, S., Jonker, W.: 
A lucky dip as a secure data store. 
\emph{Proceedings of Workshop on Information and System Security}, 2006.

\bibitem{calvino2015privacy}
 Calvi\~{n}o, A., Ricci, S., Domingo-Ferrer, J.:
Privacy-preserving distributed statistical computation to a semi-honest multi-cloud. 
\emph{In IEEE Conf. on Communications and Network Security (CNS 2015)}. IEEE 2015.

\bibitem{cao2014privacy}
Cao, N., Wang, C., Li, M., Ren, K., Lou, W.:
 Privacy-preserving multi-keyword ranked search over encrypted cloud data. 
 \emph{IEEE Transactions on parallel and distributed systems}, 25(1):222-233, 2014.

\bibitem{C} Carter, M. W.: 
A Survey of Practical Applications of Examination Timetabling Algorithms.
\emph{Operations Research} 34 (2): 193–202.

\bibitem{CirianiDFJPS:07}
Ciriani, V., De Capitani di Vimercati, S., Foresti, S., Jajodia, S., Paraboschi, S., Samarati, P.:
Fragmentation and encryption to enforce privacy in data storage.
 \emph{ In Computer Security -- ESORICS 2007,
  Lecture Notes in Computer Science}, 
  vol. 4734, pp. 171--186, Springer 2007.

\bibitem{CCFJPS}
Ciriani, V., De Capitani di Vimercati, S., Foresti, S., Jajodia, S., Paraboschi, S., Samarati, P.:
Combining fragmentation and encryption to protect privacy in data storage. 
\emph{ACM Trans. Inf. Syst. Secur.} 13(3), 2010.

\bibitem{ciriani2011selective}
Ciriani, V., De Capitani di Vimercati, S., Foresti, S., Jajodia, S., Paraboschi, S., Samarati, P.:
Selective data outsourcing for enforcing privacy.
\emph{Journal of Computer Security}, 19(3):531--566, 2011.

\bibitem{clarus}
CLARUS D3.2 or David S\'anchez and Montserrat Batet. 
``Privacy-preserving data outsourcing in the cloud via semantic data splitting''. Manuscript

\bibitem{clifton2002tools}
Clifton, C. , Kantarcioglu, M., Vaidya, J., Lin, X., Zhu, M.Y.: 
Tools for privacy preserving distributed data mining.
\emph{ACM Sigkdd Explorations Newsletter 4.2 }: 28-34, 2002.

\bibitem{cox1992ideals}
Cox, D., Little, J., O'shea, D.:
Ideals, varieties, and algorithms. 
New York: Springer, 1992.

\bibitem{de1995grobner}
De Loera, J.A.:
Gr\"{o}bner bases and graph colorings. 
\emph{Beitr\"{a}ge zur algebra und geometrie.} 36(1):89-96, 1995.


\bibitem{de2015graph}
De Loera, J.A., Margulies, S., Pernpeintner, M., Riedl, E., Rolnick, D., Spencerm G., Stasim D., Swenson, J.:
Graph-coloring ideals: Nullstellensatz certificates, Gr\"{o}bner bases for chordal graphs, and hardness of Gr\"{o}bner bases. 
\emph{In Proceedings of the 2015 ACM on International Symposium on Symbolic and Algebraic Computation}, pp. 133-140. ACM, 2015.


\bibitem{dev2012approach}
Dev, H., Sen, T., Basak, B., Ali, M.E.:
An approach to protect the privacy of cloud data from data mining based attacks. 
\emph{High Performance Computing, Networking, Storage and Analysis (SCC)}, 2012 SC Companion:. IEEE, 2012.

\bibitem{du2004privacy}
Du, W., Yunghsiang, S.H., Shigang, C.: 
Privacy-preserving multivariate statistical analysis: Linear regression and classification. 
\emph{In Proceedings of the 2004 SIAM International Conference on Data Mining}.
Society for Industrial and Applied Mathematics, 2004.

\bibitem{EGP66}
Erd\H{o}s, P., Goodman, A.W., P\'osa, L.:
The representation of a graph by set intersections.
\emph{In Can. J. Math}, 18, pp. 106--112, 1966.

\bibitem{esponda2008hiding}
Esponda, F.: 
Hiding a needle in a haystack using negative databases. 
\emph{International Workshop on Information Hiding.} Springer Berlin Heidelberg, 2008.

\bibitem{augere1993efficient}
Faugere, J.C., Gianni, P., Lazard, D., Mora, T.:
Efficient computation of zero-dimensional Gr\"obner bases by change of ordering. 
\emph{Journal of Symbolic Computation};16(4):329-44, 1993.

\bibitem{FRR}
Farr\`as, O., Ribes-Gonz\'alez, J., Ricci, S.:
Local Bounds for the Optimal Information Ratio of Secret Sharing Schemes. https://eprint.iacr.org/2016/726

\bibitem{GM-2}
Ganapathy, V., Thomas, D., Feder, T., Garcia-Molina, H., Motwani, R.:
Distributing data for secure database services.
\emph{Transactions on Data Privacy}, vol. 5, no. 1, pp. 253--272, 2012.

\bibitem{gaye2014datashield}
Gaye, A., Marcon, Y., Isaeva, J.,  LaFlamme, P.,  Turner, A.,  Jones, E.M., Minion, J., et al.: 
DataSHIELD: taking the analysis to the data, not the data to the analysis. 
\emph{International journal of epidemiology 43.6 (2014)}: 1929-1944, 2014.

\bibitem{goethals2004private}
Goethals, B., Laur, S., Lipmaa, H., Mielik\"{a}inen, T.:
On private scalar product computation for privacy-preserving data mining.
\emph{International Conference on Information Security and Cryptology.} Springer Berlin Heidelberg, 2004.

\bibitem{guruswami2002hardness}
Guruswami, V., Hastad, J., Sudan, M.: 
Hardness of approximate hypergraph coloring.
\emph{SIAM Journal on Computing. 31.6}: 1663-1686, 2002.

\bibitem{H41}
Hall, M. Jr.:
A problem in partitions.
\emph{Bull. Amer. Math. Soc.}, 47, pp. 801--807, 1941.

\bibitem{hillar2008algebraic}
Hillar, C.J., Windfeldt, T.:
Algebraic characterization of uniquely vertex colorable graphs. 
\emph{Journal of Combinatorial Theory}, Series B;98(2):400-14, 2008.

\bibitem{kantarcioglu2008survey}
Kantarcioglu, M.: 
A survey of privacy-preserving methods across horizontally partitioned data.
\emph{Privacy-preserving data mining}. Springer US, 313-335, 2008.

\bibitem{karr2009privacy}
Karr, A.F., Lin, X., Sanil, A.P., Reiter, J.P.:
Privacy-preserving analysis of vertically partitioned data using secure matrix products.
\emph{Journal of Official Statistics 25.1 (2009)}: 125, 2009.

\bibitem{L} Leighton, F. T.:
A Graph Coloring Algorithm for Large Scheduling Problems.
\emph{J. Res. Natl. Bur. Standards} 84: 489–506, 1979.

\bibitem{levy2009survey}
Levy-dit-Vehel, F., Marinari, M.G., Perret, L., Traverso, C.:
A survey on Polly Cracker systems. 
\emph{Gr\"obner Bases, Coding, and Cryptography}:285-305, 2009.

\bibitem{loera2009expressing}
Loera, J.A., Lee, J., Margulies, S., Onn, S.: 
Expressing combinatorial problems by systems of polynomial equations and hilbert's nullstellensatz. 
\emph{Combinatorics, Probability and Computing};18(4):551-82, 2009.


\bibitem{ricci2016privacy}
Ricci, S., Domingo-Ferrer, J., S\'{a}nchez, D.: 
Privacy-Preserving Cloud-Based Statistical Analyses on Sensitive Categorical Data. 
\emph{Modeling Decisions for Artificial Intelligence.} Springer International Publishing, 2016.

\bibitem{sage}
The Sage Mathematical Software System:  
\url{http://www.sagemath.org/}

\bibitem{S87}
Spencer, J.:
Ten lectures on the probabilistic method.
\emph{In SIAM Regional Conference Series in Applied Mathematics}, vol. 52, 1987.

\bibitem{sweeney2000simple}
Sweeney, L.:
Simple demographics often identify people uniquely. 
Health (San Francisco);671:1-34, 2000.

\bibitem{WP}
Welsh, D.J.A., Powell, M.B.: 
An upper bound for the chromatic number of a graph and its application to timetabling problems. 
\emph{The Computer Journal} 10 (1), 85-86, 1967.

\bibitem{wei2013data}
Wei, Z., Xinwei, S., Tao, X.:
Data privacy protection using multiple cloud storages. 
\emph{Mechatronic Sciences, Electric Engineering and Computer (MEC)}, Proceedings 2013 International Conference on. IEEE, 2013.

\bibitem{yang200610}
Yang, Q., Wu, X.:
10 challenging problems in data mining research. 
\emph{International Journal of Information Technology \& Decision Making 5.04 (2006)}: 597-604, 2006.

\end{thebibliography}
\end{document}